\documentclass{article}

\usepackage[T1]{fontenc}
\usepackage{amsmath,amsfonts,amssymb,amsthm}
\usepackage{mathrsfs}
\usepackage{xspace}
\usepackage{wrapfig}
\usepackage[font=footnotesize,labelsep=space,labelfont=bf]{caption}
\usepackage{url}
\usepackage{appendix}
\usepackage{color,graphicx}
\usepackage[pdftex,bookmarks=true,pdfstartview=FitH,colorlinks,linkcolor=blue,filecolor=blue,citecolor=blue,urlcolor=blue,pagebackref=true]{hyperref}

\usepackage{fullpage}
\usepackage[top=0.88in,bottom=0.88in,left=0.88in,right=0.88in]{geometry}

\newtheorem{theorem}{Theorem}
\newtheorem{lemma}{Lemma}

\newtheorem{claim}{Claim}
\newtheorem{observation}{Observation}

\newcommand{\namedref}[2]{\hyperref[#2]{#1~\ref*{#2}}}

\newcommand{\Sectionref}[1]{\namedref{Section}{sec:#1}}
\newcommand{\Appendixref}[1]{\namedref{Appendix}{app:#1}}
\newcommand{\Theoremref}[1]{\namedref{Theorem}{thm:#1}}

\newcommand{\Lemmaref}[1]{\namedref{Lemma}{lem:#1}}
\newcommand{\Claimref}[1]{\namedref{Claim}{clm:#1}}
\newcommand{\Figureref}[1]{\namedref{Figure}{fig:#1}}
\newcommand{\Equationref}[1]{\namedref{Equation}{eq:#1}}
\newcommand{\Eqref}[1]{\eqref{eq:#1}}
\newcommand{\Footnoteref}[1]{\namedref{Footnote}{foot:#1}}

\newcommand{\prot}{\ensuremath{\pi}\xspace}
\newcommand{\eps}{\ensuremath{\mathscr{E}}\xspace}
\newcommand{\ceps}{\ensuremath{\varepsilon}\xspace}
\newcommand{\Real}{\ensuremath{{\mathbb R}}\xspace}
\newcommand{\pr}{\operatorname{Pr}}
\newcommand{\prob}{\mathbf{p}}
\newcommand{\err}{\ensuremath{\mathrm{err}}\xspace}
\newcommand{\avgerr}{\ensuremath{\overline{\err}}\xspace}

\newcommand{\cA}{\ensuremath{\mathcal{A}}\xspace}
\newcommand{\cB}{\ensuremath{\mathcal{B}}\xspace}
\newcommand{\cD}{\ensuremath{\mathcal{D}}\xspace}
\newcommand{\cQ}{\ensuremath{\mathcal{Q}}\xspace}
\newcommand{\cS}{\ensuremath{\mathcal{S}}\xspace}
\newcommand{\cT}{\ensuremath{\mathcal{T}}\xspace}
\newcommand{\cX}{\ensuremath{\mathcal{X}}\xspace}
\newcommand{\cY}{\ensuremath{\mathcal{Y}}\xspace}
\newcommand{\cZ}{\ensuremath{\mathcal{Z}}\xspace}

\newcommand{\mininfo}{\ensuremath{I_{\infty}}\xspace}

\newcommand{\ICr}{\ensuremath{IC_{\infty}}\xspace}
\newcommand{\IC}{\ensuremath{IC}\xspace}
\newcommand{\ICint}{\ensuremath{IC^{\mathrm{int}}}\xspace}
\newcommand{\CC}{\ensuremath{R}\xspace}

\newcommand{\p}{\ensuremath{\mathrm{p}}\xspace}
\newcommand{\pIC}{\ensuremath{{\p}IC}\xspace}
\newcommand{\pICr}{\ensuremath{{\p}IC_{\infty}}\xspace}
\newcommand{\pICint}{\ensuremath{{\p}IC^{\mathrm{int}}}\xspace}

\newcommand{\xp}{\ensuremath{\mathrm{\hat{p}}}\xspace}
\newcommand{\xpIC}{\ensuremath{{\xp}IC}\xspace}
\newcommand{\xpICr}{\ensuremath{{\xp}IC_{\infty}}\xspace}
\newcommand{\xpICint}{\ensuremath{{\xp}IC^{\mathrm{int}}}\xspace}

\newcommand{\prt}{\ensuremath{\mathrm{prt}}\xspace}
\newcommand{\relprt}{\ensuremath{\overline{\mathrm{prt}}}\xspace}

\begin{document}
\title{R\'enyi Information Complexity and an Information Theoretic Characterization of the Partition Bound}
\author{Manoj M. Prabhakaran\\
Department of Computer Science\\
University of Illinois\\
 Urbana-Champaign, IL\\
	\texttt{ mmp@illinois.edu}
\and
 Vinod M. Prabhakaran\\
School of Technology and Computer Science\\
Tata Institute of Fundamental Research\\
Mumbai, India.\\
	\texttt{vinodmp@tifr.res.in}}
\maketitle
\begin{abstract}
In this work we introduce a new information-theoretic complexity measure
for 2-party functions, called R\'enyi information complexity. It is a
lower-bound on communication complexity, and has the two leading
lower-bounds on communication complexity as its natural relaxations:
(external) information complexity and logarithm of partition complexity.
These two lower-bounds had so far appeared conceptually quite different from
each other, but we show that they are both obtained from R\'enyi information
complexity using two different, but natural relaxations:

1. The relaxation of R\'enyi information complexity that yields information
complexity is to change the order of R\'enyi mutual information used in its
definition from infinity to 1.

2. The relaxation that connects R\'enyi information complexity with
partition complexity is to replace protocol transcripts used in the
definition of R\'enyi information complexity with what we term
``pseudotranscripts,'' which omits the interactive nature of a protocol, but
only requires that the probability of any transcript given inputs $x$ and
$y$ to the two parties, factorizes into two terms which depend on $x$ and
$y$ separately. While this relaxation yields an apparently different
definition than (log of) partition function, we show that the two are in
fact identical. This gives us a surprising characterization of the partition
bound in terms of an information-theoretic quantity.

We also show that if both the above relaxations are simultaneously applied
to R\'enyi information complexity, we obtain a complexity measure that is
lower-bounded by the (log of) relaxed partition complexity, a complexity
measure introduced by Kerenidis et al. (FOCS 2012). We obtain a sharper
connection between (external) information complexity and relaxed partition
complexity than Kerenidis et al., using an arguably more direct proof.

Further understanding R\'enyi information complexity (of various orders)
might have consequences for important direct-sum problems in communication
complexity, as it lies between communication complexity and information
complexity.
\end{abstract}

\section{Introduction}
\label{sec:intro}

Communication complexity, since the seminal work of Yao~\cite{Yao79}, has
been a central question in theoretical computer science. Many of the recent
advances in this area have centered around the notion of information
complexity, which measures the {\em amount of information} about the inputs
-- rather than the {\em number of bits} -- that should be present in a
protocol's transcript, if it should compute a function (somewhat) correctly.
The more traditional approach for lower bounding communication complexity
relied on {\em combinatorial complexity measures} of functions. The goal of this
work is to relate these two lines of studying communication complexity with
each other.

Currently, the two leading lower bounds for communication complexity in the
literature come from these two lines: (external) information complexity \IC
\cite{ChakrabartiShWiYa01,Bar-YossefJaKuSi04} and partition complexity \prt
\cite{JainKl10}. Either of these two lower bounds upper-bounds (and hence gives an
equally good or better lower bound than) all the other bounds used in the
literature. An intriguing problem in this area has been to understand if one
of these two bounds is a better lower-bound than the other. An important
motivation behind this problem is the possibility of separating \IC from
communication complexity via an intermediate combinatorial lower bound,
which will have consequences for direct-sum results in communication
complexity (since \IC is known to be equal to amortized communication
complexity~\cite{BravermanRa11,Braverman12}).

\begin{wrapfigure}[30]{r}{0.4\textwidth}
\centering
\includegraphics[page=1,trim=250 0 100 0,clip,width=\linewidth]{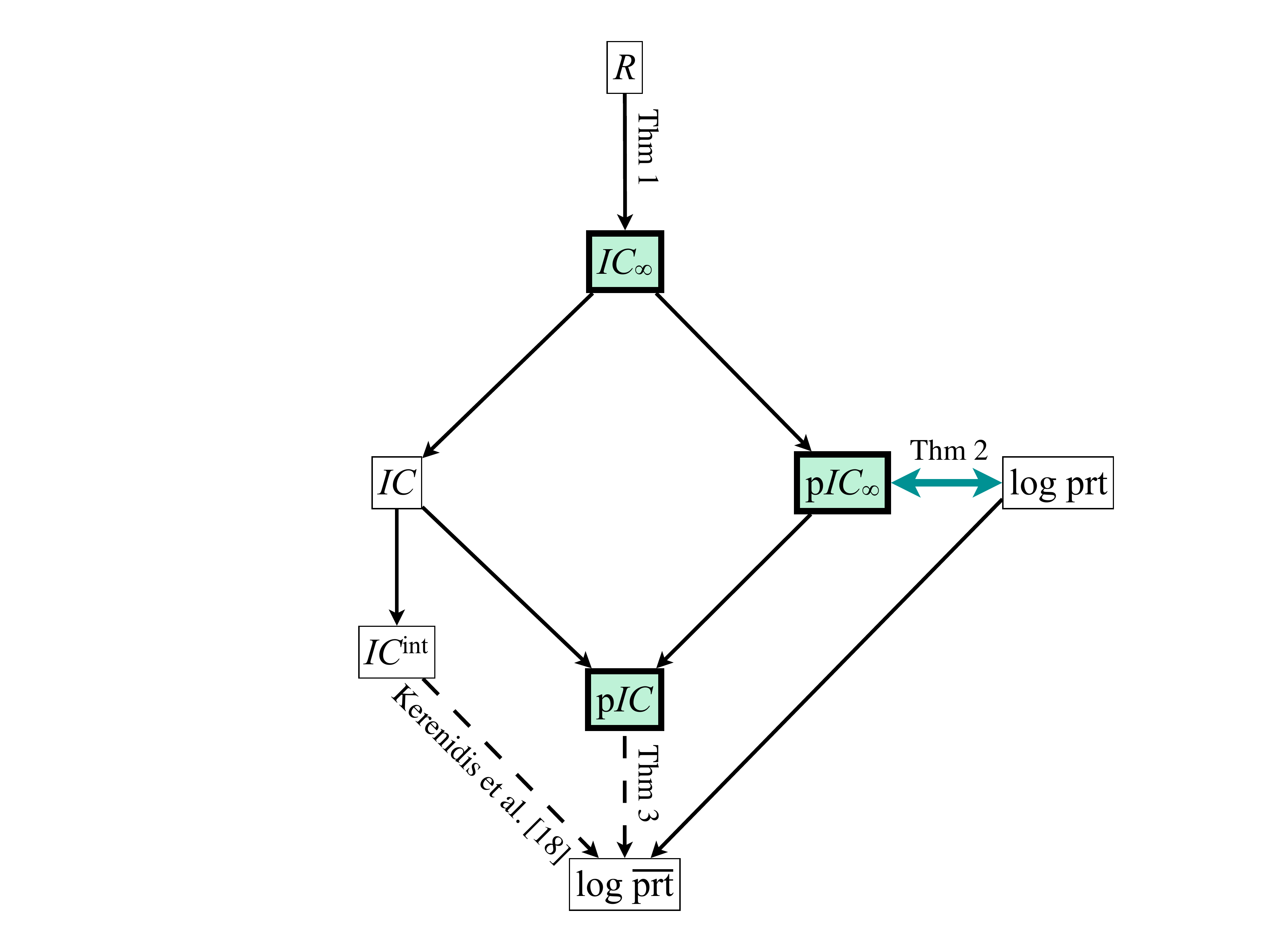}
\caption{New complexity measures (shaded) and their relation to existing
ones. Existing ones shown include the (public-coin) worst-case
communication complexity ($R$), external and internal information complexity
(\IC and \ICint), partition complexity (\prt) and relaxed partition
complexity (\relprt). An arrow from one measure to another shows that the
latter is a lower-bound for the former.  (The dashed lines indicate that the
lower bound holds up to constant factors and shifts in error bounds.) \pICr
is exactly equal to $\log\prt$.\label{fig:map}}
\end{wrapfigure}
Kerenidis et al.\ \cite{KerenidisLaLeRoXi12} showed that information
complexity ``subsumes'' (the logarithm of) a relaxed variant of partition
complexity, \relprt, in the sense that any lower bound on $\log\relprt$ in
fact yields a lower bound on information complexity. Thus bounding
$\log\relprt$ cannot yield stronger lower bounds than bounding information
complexity. In turn, all the combinatorial bounds in the literature -- other
than $\log\prt$ -- are subsumed by $\log\relprt$. On the other hand, in recent breakthrough
results, Ganor, Kol and Raz \cite{GanorKoRa14FOCS,GanorKoRa15} showed that
for a certain range of parameters, combinatorial lower bounds can be
significantly stronger than information complexity lower bounds.%
\footnote{These results use combinatorial lower bounds to 
establish that communication complexity could be
exponentially larger than information complexity. The communication
complexity in these examples is (sub-)logarithmic in the size of the input
itself.}
It remains open if such separations are possible for a less restrictive
range of parameters (e.g., with communication complexity that is say,
super-logarithmic in the input size). In the absence of a result analogous
to that of \cite{KerenidisLaLeRoXi12} for \prt itself, \prt remains a
candidate for showing such separations.

In this work, we do not pursue the question of whether $\log \prt$ could be
larger than \IC or vice versa. Instead, we develop a new
information-theoretic complexity measure, \ICr which is as large or larger
than both \IC and $\log \prt$ (see \Figureref{map}), and has {\em natural}
relaxations that yield \ICr and $\log \prt$ respectively. \ICr {\em is thus
a candidate for separating \IC and communication complexity for a larger
range of parameters than currently known to be possible}. Further, the
relaxation of \ICr to $\log \prt$ reveals a surprising information-theoretic
definition for \prt.  Since this new definition of (log of) \prt has a
markedly different form, we give it a different name, \pICr.

We also consider applying {\em both the relaxations} mentioned above
simultaneously to \ICr. This yields a new complexity measure \pIC. We then 
show that \pIC is essentially lower bounded by $\log \relprt$, the relaxed partition complexity.  {\em
This recovers a result similar to that of \cite{KerenidisLaLeRoXi12}, but
with sharper parameters and an
arguably more direct proof.}%
\footnote{Our result does not subsume the result of Kerenidis et al.\
\cite{KerenidisLaLeRoXi12}, as they deal with internal information
complexity, while it is more natural for us to work with external
information complexity. Conversely, the result of \cite{KerenidisLaLeRoXi12}
does not yield our result for external information complexity (due to the
parameters), nor the relation with the intermediate complexity measure
\pIC.}

The relation between the new and old complexity measures are shown in
\Figureref{map}. (Also see \Figureref{map-ext} for further extensions.)
The new complexity measures are informally described below.

\subparagraph*{R\'enyi Information Complexity.}
(External) Information complexity of a function is defined as the mutual
information between the transcript and the inputs, and is a lower bound
on the communication complexity of the function. The notion of mutual
information in this definition is due to Shannon. R\'enyi mutual
information $I_\alpha(A;B)$, parametrized by $\alpha \ge 0$, is a
generalization of Shannon's mutual information (see \cite{Verdu15} for a
recent treatment), with the latter corresponding to $\alpha \rightarrow 1$. 
We observe that information complexity continues to be a lower bound on
communication complexity for all values of $\alpha$. In particular, we may
consider $I_\infty$ instead of $I_1$ to define information complexity.
The resulting notion of information complexity will be called $\ICr$.

\subparagraph*{Pseudotranscript Complexity.} Communication complexity, as well
as information complexity, is defined in terms of a protocol. In contrast,
the more traditional combinatorial lower bounds on communication complexity
are defined in terms of simpler combinatorial properties of the function's
truth table. We propose complexity measures based on one such property
(which has been widely used in the analysis of protocols, but to the best of
our knowledge, has never been isolated to define a complexity measure of
functions). 

Consider a function (generalized later to relations)
$f:\cX\times\cY\rightarrow\cZ$. We define a random variable $Q$ over a space
\cQ to be a {\em pseudotranscript} for $f$ if there exist two functions
$\alpha:\cQ\times\cX\rightarrow\Real^+$ and
$\beta:\cQ\times\cY\rightarrow\Real^+$, such that $\Pr[Q=q|X=x,Y=y] =
\alpha(q,x)\beta(q,y)$, for all $q\in\cQ,x\in\cX,y\in\cY$. This definition
is motivated by the fact that the transcripts in a protocol do satisfy it
(see \Footnoteref{trans-is-pseudo}). However, a pseudotranscript need not
correspond to a protocol (indeed, any ``tiling'' of a function's table
yields a pseudotranscript, but it need not correspond to a valid protocol).
We also associate a value $z_q$ with a pseudotranscript $q$; the error
$\err_{f,Q}$ is defined in terms of the probability of this value matching
the function's output. We do not include any other properties of a protocol
in defining a pseudotranscript.

We can define complexity measures \pIC and \pICr as relaxations of \IC and
\ICr, simply by replacing protocols in their definitions with
pseudotranscripts.

\subparagraph*{Relations Among the Complexities.} The main results in this
work, apart from introducing the new complexity measures, are connections
between \pICr and \prt and between \pIC and \relprt.

\noindent $\bullet$ Firstly, we show that $\pICr=\log \prt$.
\pICr and $\prt$ are defined very differently. \prt is concerned with
{\em tiling} the function table with weighted tiles: a tile $t$ is a
rectangle in the input domain along with an output value $z_t$. \prt is the
minimum total weight of tiles needed such that for each input $(x,y)$, the
weight of the tiles covering it adds up to 1, and the weight of the tiles
with $z_t\not=f(x,y)$ is below the error threshold $\eps(x,y)$.%
\footnote{For \prt, as well as \pICr and \ICr, we use a very general notion
of error, in which the error is specified as a function
$\eps:\cX\times\cY\rightarrow[0,1]$.}
On the other hand, \pICr relates to pseudotranscripts $q$, which are similar
to tiles in that they define a value $z_q$ and a rectangle of all $(x,y)$
such that $p(q|x,y)>0$, but are more general in that there is no single
``weight'' on such a rectangle. Given our definitions, it is not hard to see
that $\log \prt$ is as large or larger than \pICr, as any tiling can be
naturally interpreted as a pseudotranscript $Q$ with the same error, and in
that case, the log of the value of the tiling indeed equals
$\mininfo(X,Y;Q)$. What is more surprising is that any pseudotranscript $Q$
can be converted to a tiling of the appropriate value (and same error). This
conversion ``slices'' an uneven weight function $p(q|x,y)$ over a rectangle
into weights $\omega_{q,t}$ over tiles $t$ inside the rectangle; the weight
of a tile $t$ is the sum of the contributions to its weight from all the
different values of $q$: $w(t)=\sum_q \omega_{q,t}$. Then it turns out that
the value of the tiling so obtained will be equal to $\mininfo(X,Y;Q)$. 

This equivalence gives a new perspective on the partition complexity.
Firstly, it shows that partition complexity exploits {\em exactly} the
properties of a pseudotranscript, which is not apparent from its original
definition. Secondly, it gives an information theoretic interpretation of a
complexity measure defined in a traditional combinatorial manner. This is
the first instance of the two lines of lower-bounding techniques for
communication complexity -- information theoretic and combinatorial --
converging. 

\noindent $\bullet$ Our second main result is that lower bounds on $\log
\relprt$ are in fact lower bounds on \pIC. More precisely, we show that
$\pIC(f,\ceps) \ge \delta\log\relprt(f,\ceps+\delta) -
(\delta\log\log|\cX||\cY| + 3)$.  This is along the same lines as the result
of \cite{KerenidisLaLeRoXi12}, with improved parameters (in
\cite{KerenidisLaLeRoXi12}, the multiplicative overhead in the leading term
is $\delta^2$ instead of $\delta$). 

The proof of this result is technically more involved, but is
closely based on the simple slicing construction from the above result. The
high-level idea is to first slice $p(q|x,y)$ into weights $\omega_{q,t}$
for each tile $t$, and then discard the contributions to $w(t)$ from
those $\omega_{q,t}$ which are too large. One needs to ensure that the
weight of the tiles discarded in this fashion is small (as it contributes
to the error), while the weight of the remaining tiles is also small (as
it contributes to the value of the tiling). For the first part, we show
how (Shannon's) mutual information $I(X,Y;Q)$ can be approximated by a convex
combination of non-negative values, and then apply Markov's
inequality. For the second part, we rely on a geometric argument to derive a
bound on the weight of the remaining tiles.

\subsection{Related Work}
Many of the recent advances in the field of communication complexity
\cite{Yao79} have followed from using various notions of information
complexity. Earlier notions of information complexity appeared implicitly
in several works \cite{Ablayev96,PonzioRaVe01,SaksSu02}, and was first
explicitly defined in \cite{ChakrabartiShWiYa01} and further developed in
\cite{Bar-YossefJaKuSi04}. Information
complexity has been extensively used or studied in the recent
communication complexity literature (e.g., 
\cite{BravermanRa11,Braverman12,BravermanWe12,ChakrabartiKoWa12,KerenidisLaLeRoXi12,BarakBrChRa13,GanorKoRa14FOCS,FontesJaKeLaLaRo15,GanorKoRa15}).
The notion was also adapted to specialized models or tasks
\cite{JayramKuSi03, JainRaSe03, JainRaSe05,HarshaJaMcRa10}.

The partition bound was developed in \cite{JainKl10}, and has subsumed a
long line of combinatorial bounds \cite{KushilevitzNi97book} (see e.g.,
\cite{JainKl10,FontesJaKeLaLaRo15}). The relaxed partition bound put
forth in \cite{KerenidisLaLeRoXi12}, similarly subsumes several
combinatorial bounds, with the exception of the partition bound itself.

In 1960, generalizing Shannon's entropy, R\'enyi proposed new measures of
entropy and divergence~\cite{Renyi60}, now known after him. Subsequently,
several authors developed different notions of mutual information based on
these measures. One such definition attributed to Sibson~\cite{Sibson69}
has recently come to be regarded as the most standard choice \cite{Verdu15},
and this is the basis for our definition of $\mininfo(A:B)$. Properties of
$I_\alpha$ for various parameters $\alpha$ have been studied in
\cite{HoVe15,Verdu15}. 
In information theory literature, the use of generalized notions of mutual
information to obtain strong lower bounds for ``one-shot'' versions of
communication problems (rather than amortized/direct-sum versions where
Shannon's mutual information is often appropriate) has a long history
starting with the work of Ziv and Zakai~\cite{ZivZa73,ZakaiZi75}. In the
communication complexity literature, R\'enyi divergence was used as a
technical tool in deriving one of the results in \cite{Bar-YossefJaKuSi04}.

Recently, the authors of this work proposed a distributional complexity
measure, {\em Wyner tension} (or more generally, tension gap) which is a
lower bound for information complexity \cite{PrabhakaranPr14a}. We leave it
for future work to explore the exact connections between these bounds and
the ones in the current work. We mention that for the case when the inputs
are independent, Wyner tension is identical to \pICint (defined in
\Sectionref{extensions}), and a result in \cite{PrabhakaranPr14a} is
subsumed by the results in this work.

\section{Preliminaries}
\label{sec:prelims}

Let $f:{\mathcal X}\times{\mathcal Y} \rightarrow 2^{\mathcal Z}$ be a
relation. Alice who has input $x\in{\mathcal X}$ and Bob who has input
$y\in{\mathcal Y}$ want to output any $z\in f(x,y)$. 
We consider public-coin protocols, in which Alice and Bob have access to
a common random string independent of the inputs; they may also use private
local randomness. For such a protocol $\prot$, we say that the probability
of error, which we view as a {\bf function} of $(x,y)\in \cX\times \cY$, is
\[ \err_{f,\prot}(x,y) = \pr[\prot(x,y)\notin f(x,y)],\] where $\prot(x,y)$ is
the {\em output} of the protocol and the probability is over the randomness in
the protocol execution.%
\footnote{For a protocol to be considered valid, we will insist that the two
parties output the same value with probability 1; hence the output of a
protocol is well-defined.}
An error function $\eps$ that is of particular interest is the constant (or
worst-case) error function: $\eps(x,y)=\ceps$ for some constant $\ceps$, for
all $(x,y)\in\cX\times\cY$.

For a protocol \prot, let $\#\text{bits}(\prot,x,y)$ denote the maximum
number of bits exchanged in an execution of \prot with inputs $(x,y)$, in
the worst case (i.e., over all choices of randomness). Note that this
measure excludes the number of bits in the public randomness. The
(worst case) {\em communication complexity} $\CC(f,\eps)$ of $f$,
for an error function $\eps$, is defined as
\[ \CC(f,\eps) = \inf_{\substack{\text{protocol } \prot:\\\err_{f,\prot}\leq \eps}}\; \max_{x,y}\quad \#\text{bits}(\prot,x,y).\]
To define information complexities, we will need to consider the
distribution $\prob_{X,Y}$ on the inputs $X,Y$. Let $\Pi$ be the
random variable that denotes 
the communication transcript {\em and} 
the public-coins 
of the protocol $\prot$.
Then, the external information cost of the protocol $\pi$ under the input distribution $\prob_{X,Y}$
is $I(X,Y;\Pi)$, i.e., 
the amount of
information about the inputs $X,Y$ contained in $\Pi$.
The (non-distributional) {\em external information complexity} $\IC(f,\eps)$ is defined as
\[ \IC(f,\eps) = \inf_{\substack{\text{protocol } \prot:\\\err_{f,\prot}\leq \eps}}\; \max_{\prob_{X,Y}}\quad I(X,Y;\Pi).\]
Similarly, internal information complexity is defined as
\[ \ICint(f,\eps) = \inf_{\substack{\text{protocol } \prot:\\\err_{f,\prot}\leq \eps}}\; \max_{\prob_{X,Y}}\quad I(X;\Pi|Y) + I(Y;\Pi|X).\]
Here the internal information cost, $I(X;\Pi|Y) + I(Y;\Pi|X)$, of the protocol $\pi$ under input distribution $\prob_{X,Y}$ is the sum of the information learned by the parties about each other's input from $\Pi$.
The following relationship between these quantities is well-known.
\[ \ICint(f,\eps) \leq \IC(f,\eps) \leq \CC(f,\eps).\]

A {\em tile} for $(\cX,\cY,\cZ)$ is a pair $(r_X\times r_Y,z)$, where $r_X
\subseteq \cX$, $r_Y \subseteq \cY$ and $z\in\cZ$. If $t= (r_X\times
r_Y,z)$, then we let $\cX_t, \cY_t$, and $z_t$ denote $r_X,r_Y$ and $z$
respectively. We say $(x,y)\in t$ if and only if $x\in \cX_t$
and $y\in \cY_t$. The set of all tiles for $(\cX,\cY,\cZ)$ is denoted by $\cT(\cX,\cY,\cZ)$ or simply $\cT$ (if $\cX,\cY,\cZ$ are clear from the context).

For a relation $f:{\mathcal X}\times{\mathcal Y} \rightarrow 2^{\mathcal Z}$
and probability of error $\eps:\cX\times \cY \rightarrow [0,1]$, the
partition complexity~\cite{JainKl10} is defined as follows:%
\footnote{The definition presented in \cite{JainKl10} is slightly more
restrictive in the kind of relations and error functions considered.}
\begin{align}
\prt(f,\eps) = \min_{w:\cT\rightarrow[0,1]} \sum_{t\in \cT}\; w(t)\qquad&\text{subject to} \notag\\
\sum_{t\in \cT: (x,y)\in t} w(t) &= 1, &&\forall (x,y)\in \cX\times\cY
\label{eq:prt-total}\\
\sum_{\substack{t\in \cT: (x,y)\in t,\\z_t \in f(x,y)}} w(t) &\geq 1 - \eps(x,y), &&\forall (x,y)\in \cX\times\cY.
\label{eq:prt-err}
\end{align}
For a weight function $w$ as above, we define the error function as
$\err_{f,w}(x,y) = \sum_{\substack{t\in \cT: (x,y)\in t,\\z_t \notin f(x,y)}} w(t)$;
then the condition \Eqref{prt-err} can be written as a condition on this error
function: $\err_{f,w} \le \eps$.

The relaxed partition complexity~\cite{KerenidisLaLeRoXi12} relaxes the equality constraint in \Eqref{prt-total} to an inequality. Further, the error function is restricted to be a constant function given by
$\eps(x,y)=\ceps$. Specifically, for a relation $f$ and a constant $0 \le \ceps \le 1$,
\begin{align}
\relprt(f,\ceps) = \min_{w:\cT\rightarrow[0,1]} \sum_{t\in \cT}\; w(t)\quad&\text{subject to}\notag\\
\sum_{t\in \cT: (x,y)\in t} w(t) &\le 1, &&\forall (x,y)\in \cX\times\cY
\label{eq:relprt-total}\\
\sum_{\substack{t\in \cT: (x,y)\in t,\\z_t \in f(x,y)}} w(t) &\geq 1 - \ceps, &&\forall (x,y)\in \cX\times\cY.
\label{eq:relprt-err}
\end{align}
The distributional form of relaxed partition complexity is defined for a distribution $\mu$ and $\ceps\in[0,1]$ as follows:
\begin{align*}
\relprt^\mu(f,\ceps) = \min_{w:\cT\rightarrow[0,1]}\quad &\sum_{t\in \cT} w(t)\quad
\text{subject to}
&\begin{aligned}
\forall (x,y)\in \cX\times\cY \;\; \sum_{t\in \cT: (x,y)\in t} w(t) &\leq 1, \\
\sum_{x,y} \mu(x,y) \sum_{\substack{t\in \cT: (x,y)\in t,\\z_t \in f(x,y)}} w(t) &\geq 1 - \ceps.
\end{aligned}
\end{align*}
For a weight function $w$ as above and a distribution $\mu$ over
$\cX\times\cY$, we write $\avgerr^\mu_{f,w}$ for
$1-\sum_{x,y} \mu(x,y) \sum_{\substack{t\in \cT: (x,y)\in t,\\z_t \in f(x,y)}} w(t)$; so the
second condition can be written as $\avgerr^\mu_{f,w} \le \ceps$.
As shown in \cite{KerenidisLaLeRoXi12}, $\relprt(f,\ceps) = \max_\mu \relprt^\mu(f,\ceps)$.

\section{R\'enyi Information Complexity and Pseudotranscripts}
\label{sec:definitions}
In this section we define our new complexity measures.

\subparagraph*{R\'enyi information complexity.}
For a pair of random variables $(A,B)$ over $\cA\times\cB$,
R\'enyi mutual information of order~$\infty$ is defined as (see, e.g.,~\cite{Verdu15})
\[ \mininfo(A;B) = \log \left( \sum_{b\in\cB}\; \max_{a\in\cA: \prob_A(a) > 0}\, \prob_{B|A}(b|a) \right). \] 

For a protocol $\prot$ and an input distribution $\prob_{X,Y}$, we will call
$\mininfo(X,Y;\Pi)$ the R\'enyi information cost.  {\em R\'enyi information
complexity} $\ICr(f,\eps)$ is defined as the smallest worst-case (over input
distributions) R\'enyi information cost of any protocol which has a
probability of error at most $\eps(x,y), x\in\cX, y\in\cY$.
\begin{align*}
\ICr(f,\eps) 
= \inf_{\substack{\text{protocol } \prot:\\\err_{f,\prot}\leq \eps}}\; \max_{\prob_{X,Y}} \; \mininfo(X,Y;\Pi).
\end{align*}
Note the above definition is identical to the definition of $\IC(f,\eps)$
except that $\mininfo$ is used in place of mutual information $I$. It is
easy to see that the inner maximization above is obtained by any input
distribution $\prob_{X,Y}$ with full support. Hence, we may equivalently
write
\begin{align*}
\ICr(f,\eps) = \inf_{\substack{\text{protocol } \prot:\\\err_{f,\prot}\leq \eps}} \mininfo(X,Y:\Pi),
\end{align*}
where we define $\mininfo(A:B)$ which is a function only of $\prob_{B|A}$ as
\[\mininfo(A:B)= \log \left( \sum_{b\in\cB}\; \max_{a\in\cA}\, \prob_{B|A}(b|a) \right).\]

\begin{theorem}
$\IC(f,\eps) \leq \ICr(f,\eps) \leq \CC(f,\eps)$.
\end{theorem}
\begin{proof}
The inequality $\IC(f,\eps) \leq \ICr(f,\eps)$ follows from $I(X,Y;\Pi)\leq
\mininfo(X,Y;\Pi)$, which in turn follows from the monotonicity of
$\alpha$-mutual information~\cite[Theorem~4(b)]{HoVe15}; for completeness, we give a proof that $I(A;B)\leq \mininfo(A;B)$ in the \Appendixref{renyi-monotonicity}.

The proof of $\ICr(f,\eps) \leq \CC(f,\eps)$ is simple.
Consider any public-coin protocol $\prot$. Let $\Pi=(\Phi,\Psi)$ where
$\Phi$ represents the public-coins and $\Psi$ the transcript of $\prot$.
W.l.o.g., $\Psi$ can be considered to be a deterministic function of $\Phi$
and the inputs $X,Y$.%
\footnote{Any protocol using private randomness can be transformed to one with
only public randomness, by including the private coins as part of the
public-coins, without changing the number of bits communicated. Further,
this can only increase the quantity $\mininfo(X,Y;\Pi)$.
Hence, it is enough to prove the inequality after carrying out this
transformation.}
We write $\Psi(x,y;\phi)$ to denote the transcript of $\pi$ on inputs
$(x,y)$ and public coins $\phi$. Note that $\#\text{bits}(\prot,x,y) =
\max_\phi |\Psi(x,y;\phi)|$ (where $\mid\cdot\mid$ denotes the length of a
bit string). We shall show that
$\mininfo(X,Y:\Pi) \le \max_{x,y,\phi} |\Psi(x,y;\phi)|$. 
This suffices since
\begin{align*}
\ICr(f,\eps) = \inf_{\substack{\text{protocol }\prot:\\\err_{f,\prot}\leq \eps}} \mininfo(X,Y:\Pi).
&&
\CC(f,\eps) = 
 \inf_{\substack{\text{protocol }\prot:\\\err_{f,\prot}\leq \eps}} \;\max_{x,y,\phi} |\Psi(x,y;\phi)|.
\end{align*}
Note that $\prob_{\Phi \Psi|XY}(\phi,\psi|x,y)=\prob_\Phi(\phi)\prob_{\Psi|\Phi XY}(\psi|\phi,x,y)$. Then,
\begin{align*}
\mininfo(X,Y:\Phi,\Psi)
 &= \log  \sum_\phi  \prob_\Phi(\phi) \sum_{\psi} \max_{x,y} \prob_{\Psi|\Phi XY}(\psi|\phi,x,y) \\
 &\leq \log  \max_\phi \sum_{\psi} \max_{x,y} \prob_{\Psi|\Phi XY}(\psi|\phi,x,y) \\
 &= \max_\phi \;\log  |\{\psi: \exists (x,y) \text{ s.t. } \psi = \Psi(x,y;\phi) \}|  
 \leq \max_{x,y,\phi} |\Psi(x,y;\phi)|. \qedhere
\end{align*}
\end{proof}

\subparagraph*{Pseudotranscript and pseudo-information complexities.}
A random variable  $Q$ defined on an alphabet $\cQ$ and jointly distributed with the inputs $X,Y$ is said to be a {\em pseudotranscript} if $\prob_{Q|X,Y}$ satisfies the following {\em factorization condition:} 
\[\prob_{Q|X,Y}(q|x,y) = \alpha(q,x)\beta(q,y), \quad\forall q\in\cQ,x\in\cX,y\in\cY,\]
for some pair of functions $\alpha : \cQ\times \cX \rightarrow \Real^+$
and $\beta : \cQ\times \cY \rightarrow \Real^+$.
In addition, we will require that $Q$ defines an output, i.e., for each $q$ there is an associated $z_q\in \cZ$.

For any protocol $\pi$, clearly, $\Pi$, which is composed of the
public-coins and the transcript, is a pseudotranscript.%
\footnote{\label{foot:trans-is-pseudo}%
$Q=\Pi$ satisfies the factorization condition, as in that case,
for $q=(\phi,m_1,\cdots,m_t)$, $\pr[q|x,y] = \alpha(q,x)\cdot\beta(q,y)$, where say,
$\alpha(q,x) = \pr[\phi]\cdot\Pi_{\mathrm{odd } i} \Pr[m_i \mid \phi,m_1,\cdots,m_{i-1},x]$,
and $\beta(q,y) = \Pi_{\mathrm{even } i} \Pr[m_i \mid \phi,m_1,\cdots,m_{i-1},y]$.
Also, we can associate the output of the protocol, which we insisted must be
the same for both parties for a valid protocol, as the corresponding output
$z_Q$. Though the output of the parties could in principle depend on the
local input and local randomness, the factorization condition and the
requirement that the outputs agree together imply that the output can be
unambiguously determined from the transcript together with the
public-coins.} 
For a pseudotranscript $Q$, the probability of error is defined analogously to that for a protocol as
\[\err_{f,Q}(x,y)=\pr[z_Q\notin f(x,y)|(X,Y)=(x,y)].\]
We define the following ``pseudo-quantities'' corresponding to $\ICr$ and $\IC$ where $\Pi$ is replaced by pseudotranscripts:
\begin{align*}
\pICr(f,\eps)&=\inf_{\substack{\text{pseudotranscript } Q:\\\err_{f,Q}\leq \eps}}  \max_{\prob_{X,Y}} \mininfo(X,Y;Q)
 = \inf_{\substack{\text{pseudotranscript } Q:\\\err_{f,Q}\leq \eps}} \mininfo(X,Y:Q)\\
\pIC(f,\eps)&=\inf_{\substack{\text{pseudotranscript } Q:\\\err_{f,Q}\leq \eps}} \max_{\prob_{X,Y}} I(X,Y;Q).
\end{align*}
\begin{observation}
\label{obs:pIC}
Since, for any protocol, its transcript is a pseudotranscript as well,
we have $\pICr(f,\eps) \leq \ICr(f,\eps)$ and $\pIC(f,\eps) \leq \IC(f,\eps)$.
Furthermore, since $I(A;B)\leq \mininfo(A;B)$, we also have $\pIC(f,\eps) \leq \pICr(f,\eps)$.
\end{observation}

\section{\pICr Equals the Partition Bound}
\label{sec:pICr-prt}

\begin{theorem}
\label{thm:pICr-prt}
For any relation $f:\cX\times\cY\rightarrow 2^\cZ$ and error function
\eps, $\pICr(f,\eps) = \log \prt(f,\eps)$.
\end{theorem}

We prove $\pICr(f,\eps) \le \log \prt(f,\eps)$ and $\pICr(f,\eps) \ge \log
\prt(f,\eps)$ separately. The first direction is easy, and follows by considering
the tiles in a given partition as the pseudo transcripts.
\begin{lemma}
\label{lem:pICr-le-prt}
$\pICr(f,\eps) \le \log \prt(f,\eps)$.
\end{lemma}
The proof of this lemma is given in \Appendixref{pICr-le-prt}.
Now we turn to the other direction.
\begin{lemma}
\label{lem:pICr-ge-prt}
$\pICr(f,\eps) \ge \log \prt(f,\eps)$.
\end{lemma}
The proof of this lemma will also serve as a starting point in proving the
result in \Sectionref{pIC-relaxprt}.
\begin{proof}
Suppose $\prob_{Q|X,Y}$ satisfies the factorization and output consistency
conditions, $\err_{f,Q} \le \eps$ and $\pICr(f,\eps)=\mininfo(X,Y:Q)$.  Let
$\cT$ be the set of all tiles.  To define the partition
$w:\cT\rightarrow[0,1]$, we shall (in \Eqref{omegaqt}) define quantities
$\omega_{q,t}$ (for $(q,t)\in\cQ\times\cT$) and probability distribution
$\prob_{T|Q,X,Y}$, where $T$ is a random variable over $\cT$, such that the
following conditions hold.
\begin{align}
\omega_{q,t} &= 0 & \forall (q,t) \in \cQ\times\cT \text{ s.t. } z_t \not= z_q 
\label{eq:omegaqt-zero} \\
p(q,t|x,y) &=
	\begin{cases}
	\omega_{q,t} & \text{if } (x,y)\in t \\
	0 & \text{otherwise}
	\end{cases}
	& \forall (q,t)\in\cQ\times\cT, (x,y)\in\cX\times\cY
	\label{eq:omegaqt-pqtxy}\\
\log \sum_{q\in\cQ,t\in\cT} \omega_{q,t} &= \mininfo(X,Y:Q) 
	\label{eq:omegaqt-mininfo}
\end{align}
Now, if we let $w:\cT\rightarrow[0,1]$ be defined by 
$w(t)=\sum_{q\in\cQ} \omega_{q,t}$, then
it is easy to verify that \Eqref{prt-total} and \Eqref{prt-err} hold,
and further 
$\log \prt(f,\eps) \le \log \sum_{t\in\cT} w(t) = \mininfo(X,Y:Q)
= \pICr(f,\eps)$.

Thus, to complete the proof, it suffices to define $\prob_{T|Q,X,Y}$ and
$\omega_{q,t}$ so that the above conditions
\eqref{eq:omegaqt-zero}-\eqref{eq:omegaqt-mininfo} are satisfied. Recall
that, since $Q$ is a pseudotranscript, $\prob_{Q|X,Y}$ satisfies the
factorization condition; i.e., we can write
\[\prob_{Q|X,Y}(q|x,y) = \alpha(q,x)\beta(q,y), \quad\forall q\in\cQ,x\in\cX,y\in\cY,\]
for some pair of functions $\alpha : \cQ\times \cX \rightarrow \Real^+$
and $\beta : \cQ\times \cY \rightarrow \Real^+$.
For $q\in\cQ$ and $t\in\cT$, let
\begin{align*}
\sigma_{q,t} = \min_{x\in \cX_t}\alpha(q,x)-\max_{x'\not\in \cX_t} \alpha(q,x')
\quad\text{ and }\quad
\tau_{q,t} = \min_{y\in \cY_t}\beta(q,y)-\max_{y'\not\in \cY_t} \beta(q,y').
\end{align*}
Above, in defining $\max_{x'\not\in \cX_t}$, if no such $x'$ exists -- i.e.,
$\cX_t=\cX$ -- we take the maximum to be 0 (and similarly for $\max_{y'\not\in \cY_t}$).
Now, let
\begin{equation}
\label{eq:omegaqt}
\begin{aligned}
\cT_q &= \{ t\in\cT \mid \sigma_{q,t}>0, \tau_{q,t}>0, \text{ and } z_q=z_t \} \\
\omega_{q,t} &= 
\begin{cases}
\sigma_{q,t}\cdot\tau_{q,t} &\text{if } t\in\cT_q\\
0 &\text{if } t\not\in\cT_q.
\end{cases}\\
p(t|x,y,q) &= 
\begin{cases}
{\sigma_{q,t}\cdot\tau_{q,t}}\cdot\frac1{p(q|x,y)}
	& \text{ if } (x,y)\in t, t\in\cT_q\\
0	& \text{ otherwise.}
\end{cases}
\end{aligned}
\end{equation}

We shall use the following claim (proven below) to show that $\prob_{T|X,Y,Q}$ is a valid probability distribution and that
conditions \Eqref{omegaqt-zero}-\Eqref{omegaqt-mininfo} are satisfied.
\begin{claim} 
\label{clm:telescope}
For any $q\in\cQ$ and $(x,y)\in\cX\times\cY$,
$\sum_{t\in\cT_q:(x,y)\in t} \sigma_{q,t} \cdot \tau_{q,t} = p(q|x,y)$.
\end{claim}

To see that $\prob_{T|X,Y,Q}$ is a valid
probability distribution, firstly we note that the quantity $p(t|x,y,q)$
in \Eqref{omegaqt} is well-defined: if $t\in\cT_q$ and $(x,y)\in t$, then
$\sigma_{q,t}>0,\tau_{q,t}>0$ and hence, $p(q|x,y)=\alpha(q,x)\beta(q,y) >
0$. Secondly, from \Claimref{telescope} it follows
that $\sum_{t\in\cT} p(t|x,y,q) = 1$.

Next, we verify the conditions \Eqref{omegaqt-zero}-\Eqref{omegaqt-mininfo}.
\Eqref{omegaqt-zero} directly follows from the definition of $\omega_{q,t}$.
To see \Eqref{omegaqt-pqtxy}, we note that
\begin{align*}
p(q,t|x,y)= p(q|x,y)\cdot p(t|q,x,y) 
= \begin{cases}
\sigma_{q,t}\cdot\tau_{q,t} & \text{ if } (x,y)\in t, t\in\cT_q \\
0 & \text{ if } (x,y)\in t, t\not\in\cT_q \\
0 & \text{ if } (x,y)\not\in t
\end{cases}
= \begin{cases}
\omega_{q,t} & \text{ if } (x,y)\in t \\
0 & \text{ if } (x,y)\not\in t
\end{cases}
\end{align*}

To see that \Eqref{omegaqt-mininfo} holds, 
fix a $q\in\cQ$.
Note that any $t\in\cT_q$, if
$\sigma_{q,t}\cdot\tau_{q,t} > 0$, then 
from the definition of 
$\sigma_{q,t}$ and $\tau_{q,t}$ it follows that
$(x^*,y^*)\in t$,
where $x^* = \arg\max_{x\in\cX} \alpha(q,x)$ and
$y^* = \arg\max_{y\in\cY} \beta(q,y)$. Hence
\[ \sum_{t\in\cT} \omega_{q,t} = \sum_{t\in\cT_q:(x^*,y^*)\in t}
\sigma_{q,t}\cdot\tau_{q,t} = p(q|x^*,y^*), \]
where the last equality follows from \Claimref{telescope}.
But,  $p(q|x^*,y^*) = \max_{x\in\cX,y\in\cY} \alpha(q,x)\beta(q,y) = \max_{(x,y)\in\cX\times\cY} p(q|x,y)$.
Thus, 
\[ \log \sum_{q\in\cQ,t\in\cT} \omega_{q,t} 
= \log \sum_{q\in\cQ} \max_{(x,y)\in\cX\times\cY} p(q|x,y) 
= \mininfo(X,Y:Q).\]
\end{proof}

To complete the proof of \Lemmaref{pICr-ge-prt}, we prove \Claimref{telescope}.
\begin{proof}[Proof of \Claimref{telescope}]
Fix $q\in\cQ$.  Let $\cX = \{ x_1,\cdots,x_M \}$, such that $\alpha(q,x_i)
\ge \alpha(q,x_{i-1})$ for all $i\in[1,M]$; for notational convenience,
we also define a dummy $x_{0}$ with $\alpha(q,x_0)=0$.  Define
$y_0,y_1,\cdots,y_N$ similarly for $\beta$, where $N=|\cY|$. 
Let $t_{ij} = (\cX_i \times \cY_j,z_q)$ for
$(i,j)\in[M]\times[N]$, where $\cX_i = \{x_i,\cdots,x_M\}$,
$\cY_j=\{y_j,\cdots,y_N\}$. Then,
\[\cT_q = \{ t_{ij} \mid (i,j)\in[M]\times[N], \alpha(q,x_i) >
\alpha(q,x_{i-1}), \beta(q,y_j) > \beta(q,y_{j-1}) \}. \]

Consider an arbitrary $(x,y)\in\cX\times\cY$. Let $(i^*,j^*)$ be
indices such that $(x,y)=(x_{i^*},y_{j^*})$ in the above ordering.
Note that $(x_{i^*},y_{j^*})\in t_{ij}$ if and only if $1 \le i \le i^*$
and $1\le j \le j^*$. 
Also notice that for 
all $(i,j)\in[M]\times[N]$, if $t_{ij} \not\in \cT_q$, then
$\sigma_{q,t_{ij}}, \tau_{q,t_{ij}} = 0$.
\begin{align*}
\sum_{t \in \cT_q: (x_{i^*},y_{i^*})\in t} \sigma_{q,t}\cdot\tau_{q,t} 
&= \sum_{i=1}^{i^*} \sum_{j=1}^{j^*} \sigma_{q,t_{ij}}\cdot\tau_{q,t_{ij}} \\
&= \sum_{i=1}^{i^*}
	\left({\alpha(q,x_i)} - {\alpha(q,x_{i-1})}\right) \cdot 
	\sum_{j=1}^{j^*}
	\left({\beta(q,y_j)} - {\beta(q,y_{j-1})}\right) \\
&= {\alpha(q,x_{i^*})} \cdot {\beta(q,y_{j^*})} = p(q|x_{i^*},y_{j^*})
\end{align*}
as was required to prove.
\end{proof}

\section{\pIC Subsumes Relaxed Partition Bound}
\label{sec:pIC-relaxprt}

\begin{theorem}
\label{thm:pIC-relaxprt}
For any relation $f:\cX\times\cY\rightarrow 2^\cZ$ and
constants $\ceps,\delta\in[0,1]$, 
\[\pIC(f,\ceps) \ge \delta \log \relprt(f,\ceps+\delta) - (\delta \log\log(|\cX||\cY|) + 3).\]
\end{theorem}

We prove this theorem in \Appendixref{pIC-relaxprt}. 
Below we summarize the main ideas.

\begin{proof}[Proof sketch]
It is enough to show that, given a distribution $\prob_{XY}=\mu$ over $\cX\times\cY$, 
and pseudotranscript $Q$ such that $\err_{f,Q} \le \ceps$,
there is a partition which demonstrates that
$\log \relprt^\mu(f,\ceps+\delta) \lesssim I(X,Y;Q)/\delta$.

The proof uses the construction from the proof of \Lemmaref{pICr-ge-prt},
and modifies it carefully. Specifically, we define $\prob_{T|Q,X,Y}$ and
$\omega_{q,t}$ as in \Equationref{omegaqt}.
Recall that we originally defined $w$ as
$w(t) = \sum_{q\in\cQ} \omega_{q,t}$. Our plan now
is to remove some of the weight on the tiles so that the log of the sum can
be bounded by (roughly) $I(X,Y;Q)/\delta$ as opposed to $\mininfo(X,Y:Q)$.
Towards this, we shall define a set \cB of ``bad'' pairs
$(q,t)\in\cQ\times\cT$ whose weights $\omega_{q,t}$ will not be counted
towards our new weight function $w'(t)$:
\begin{align*}
w'(t) &= \sum_{(q,t)\in(\cQ\times\cT) \setminus \cB} \omega_{q,t},\qquad  \forall t\in\cT.
\end{align*}
The crux of the proof is to define the set \cB such that the weight removed
$\sum_{(q,t)\in\cB} p(q,t)$ is below $\delta$ (it manifests as the increase
in error), while keeping $\sum_{(q,t)\notin\cB} \omega_{q,t}$
(approximately) below $I(X,Y;Q)/\delta$. We show that the following choice
of \cB has both these properties:
\begin{align*}
\cB &= \{ (q,t) \in \cQ \times \cT \mid \hat\alpha(q,t)\cdot\hat\beta(q,t) \ge \theta_q, \}
\end{align*}
where $\hat\alpha(q,t) = \min_{(x,y)\in t} \alpha(q,x)$ and 
$\hat\beta(q,t) = \min_{(x,y)\in t} \beta(q,y)$ and $\theta_q$ is an
appropriately defined threshold for each $q\in\cQ$ 
(specifically, $\theta_q = p(q)2^\Delta$, where $\Delta \approx I(XY;Q)/\delta$).

To upper bound the mass removed, we first write $ I(XY;Q) =
\sum_{q\in\cQ,t\in\cT} p(q,t) \varphi(q,t)$, where $\varphi(q,t)$ is a
quantity that is lower bounded by $\Delta$ for all $(q,t)\in\cB$.  This
suggests the possibility of using the Markov inequality to upper bound
$\sum_{(q,t)\in\cB} p(q,t)$. However, $\varphi(q,t)$ could be negative, and
we cannot directly use the above expression for $I(X,Y;Q)$ in a Markov
inequality. However, we show that removing the negative terms from
$\sum_{q,t} p(q,t)\varphi(q,t)$ does not increase the sum significantly,
which will let us still apply the Markov inequality.

To upper bound $\sum_{(q,t)\notin\cB} \omega_{q,t}$, we use a geometric
interpretation of $\omega_{q,t}$ and the set \cB.  Fix a $q\in\cQ$.  Then,
using the notation in the proof of \Claimref{telescope}, for each
$(i,j)\in[M]\times[N]$, the tile $t_{ij}$ will be represented by an
axis-parallel rectangle on the real plane, $R_{ij}$, as follows.  $R_{ij}$
is defined by its diagonally opposite vertices
$(\alpha(q,x_{i-1}),\beta(q,y_{j-1}))$ and
$(\alpha(q,x_{i}),\beta(q,y_{j}))$. (See \Figureref{tilebound}.) $R_{ij}$
could have zero area.  These rectangles tile a rectangular region, without
overlapping with each other.  Further the area of the rectangle $R_{ij}$ is
the same as $\omega_{q,t_{ij}}$. Thus $\sum_{t:(q,t)\notin\cB} \omega_{q,t}$
is given by the sum of the areas of the rectangles $R_{ij}$ for which
$(q,t_{ij}) \notin \cB$. The rectangles $R_{ij}$ that correspond to
$(q,t_{ij}) \notin \cB$ are those which have their top-right vertex (i.e.,
$(\alpha(q,x_i),\beta(q,y_j))$) fall ``below'' the hyperbola defined by the
equation $xy=\theta_q$. Thus if $(q,t_{ij})\notin\cB$, then the entire
rectangle $R_{ij}$ is below the hyperbola $xy=\theta_q$. Hence the sum of
their areas is upper-bounded by the area within $R$ that is under this
hyperbola, where $R$ is the rectangle with diagonally opposite vertices
$(0,0)$ and $(\max_{x\in\cX} \alpha(q,x), \max_{y\in\cY} \beta(q,y))$. A
calculation yields the required bound.
\end{proof}

\begin{figure}[t]
\begin{center}
\includegraphics[page=3,trim=0 100 100 225,clip,width=0.75\linewidth]{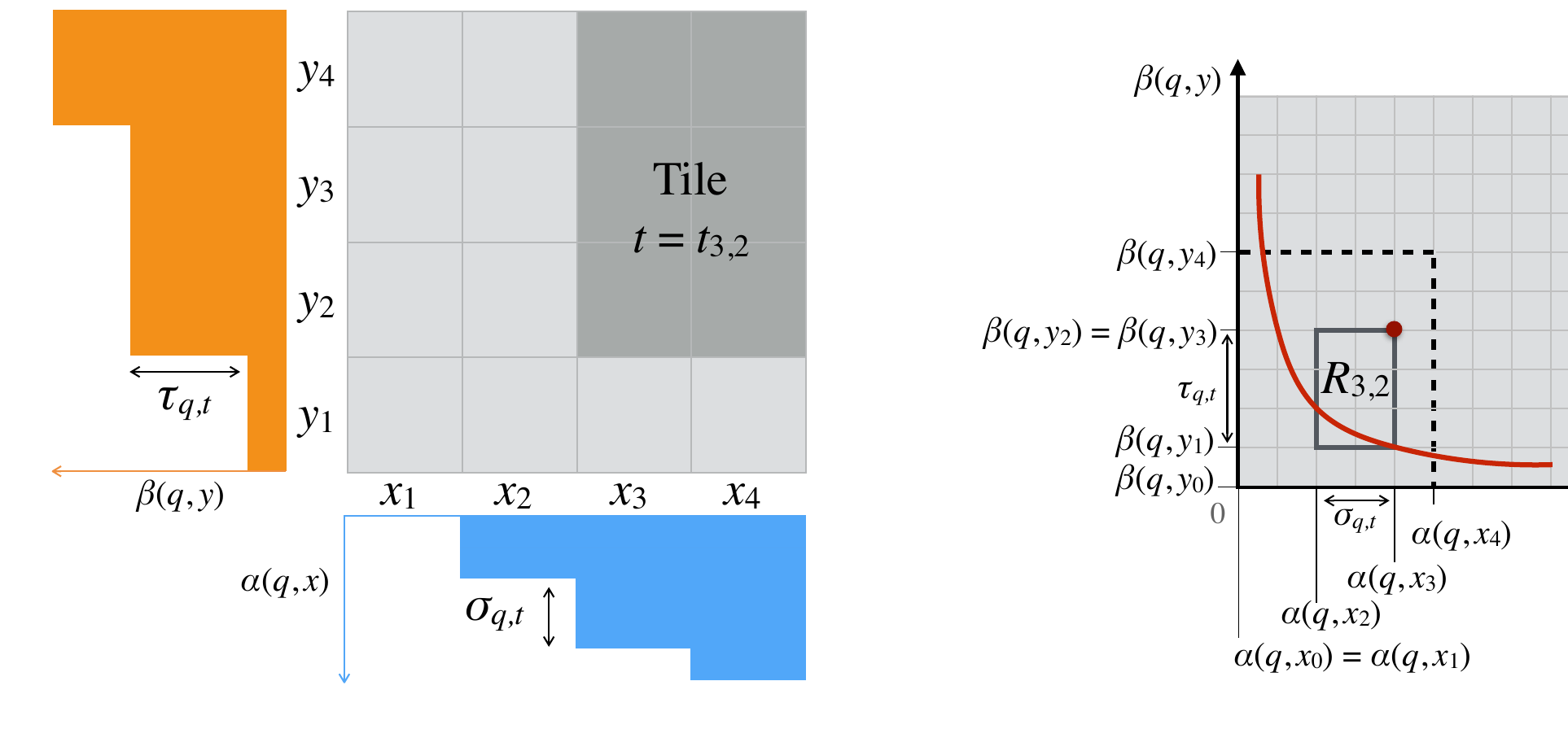}
\caption{Illustration of the geometric interpretation of \cB used in the
proof of \Theoremref{pIC-relaxprt}.
The left figure shows the domain $\cX\times\cY$ and plots
$\alpha(q,x)$ and $\beta(q,y)$ against $x$ and $y$,
which are sorted in the order of increasing $\alpha(q,x)$ and
$\beta(q,y)$, respectively
(for some fixed $q$). It also shows a tile $t=t_{3,2}$ in $\cT_q$,
and indicates the values $\sigma_{q,t}$ and $\tau_{q,t}$.
The right figure shows the alternate representation of the tile $t_{3,2}$
using the
rectangular region $R_{3,2}$. The area of $R_{3,2}$ equals
$\omega_{q,t_{3,2}}=\sigma_{q,t_{3,2}}\cdot\tau_{q,t_{3,2}}$.
A hyperbola corresponding to a threshold $\theta_q$ is also shown. 
Since the upper-right vertex of $R_{3,2}$, namely the
point $(\alpha(q,x_3),\beta(q,y_2))$ is above the hyperbola,
$(q,t_{3,2})\in\cB$. The area
within the dotted rectangle that is under the hyperbola gives an upper-bound
on the sum of areas of all rectangles under the hyperbola.
\label{fig:tilebound}}
\end{center}
\end{figure}

\section{Extensions}
\label{sec:extensions}

We may define internal information complexity associated with pseudotranscripts as
\[\pICint(f,\eps)=\inf_{\substack{\text{pseudotranscript } Q:\\\err_{f,Q}\leq \eps}} \max_{\prob_{X,Y}} I(X;Q|Y)+I(Y;Q|X).\]
It is easy to show that for the usual notion of information complexity (defined with respect to protocols), $\ICint(f,\eps)\leq \IC(f,\eps)$. The proof hinges on the fact that for any protocol $\pi$ and distribution $\prob_{X,Y}$ on the inputs, the resulting $\Pi$ satisfies the condition $I(X;Y)\geq I(X;Y|\Pi)$.
However, it is unclear whether $\pICint(f,\eps)$ is necessarily upper bounded by $\pIC(f,\eps)$. Below we define a slightly refined notion of pseudotranscripts so that information complexities defined with respect to that maintain the above inequality.

\subparagraph*{Refined pseudotranscripts and corresponding information complexities.}
A pseudotranscript $Q$ given by $\prob_{Q|X,Y}$ is called a {\em refined pseudotranscript} if, for
any distribution $\prob_{X,Y}$ on the inputs, it holds that $I(X;Y) \geq I(X;Y|Q)$.
It is easy to show that for any protocol $\pi$ and distribution $\prob_{X,Y}$ on the inputs, the resulting $\Pi$ satisfies the above condition and, hence, $\Pi$ is a refined pseudotranscript. 

Analogous to our definition of pseudo-information complexities, we define information complexities with respect to refined pseudotranscripts
\begin{align*}
\xpICr(f,\eps)&=\inf_{\substack{\text{refined pseudotranscript } Q:\\\err_{f,Q}\leq \eps}} \mininfo(X,Y:Q)\\
\xpIC(f,\eps)&=\inf_{\substack{\text{refined pseudotranscript } Q:\\\err_{f,Q}\leq \eps}} \max_{\prob_{X,Y}} I(X,Y;Q)\\
\xpICint(f,\eps)&=\inf_{\substack{\text{refined pseudotranscript } Q:\\\err_{f,Q}\leq \eps}} \max_{\prob_{X,Y}} I(X;Q|Y)+I(Y;Q|X).
\end{align*}

\begin{figure}
\begin{center}
\includegraphics[page=2,width=0.75\linewidth]{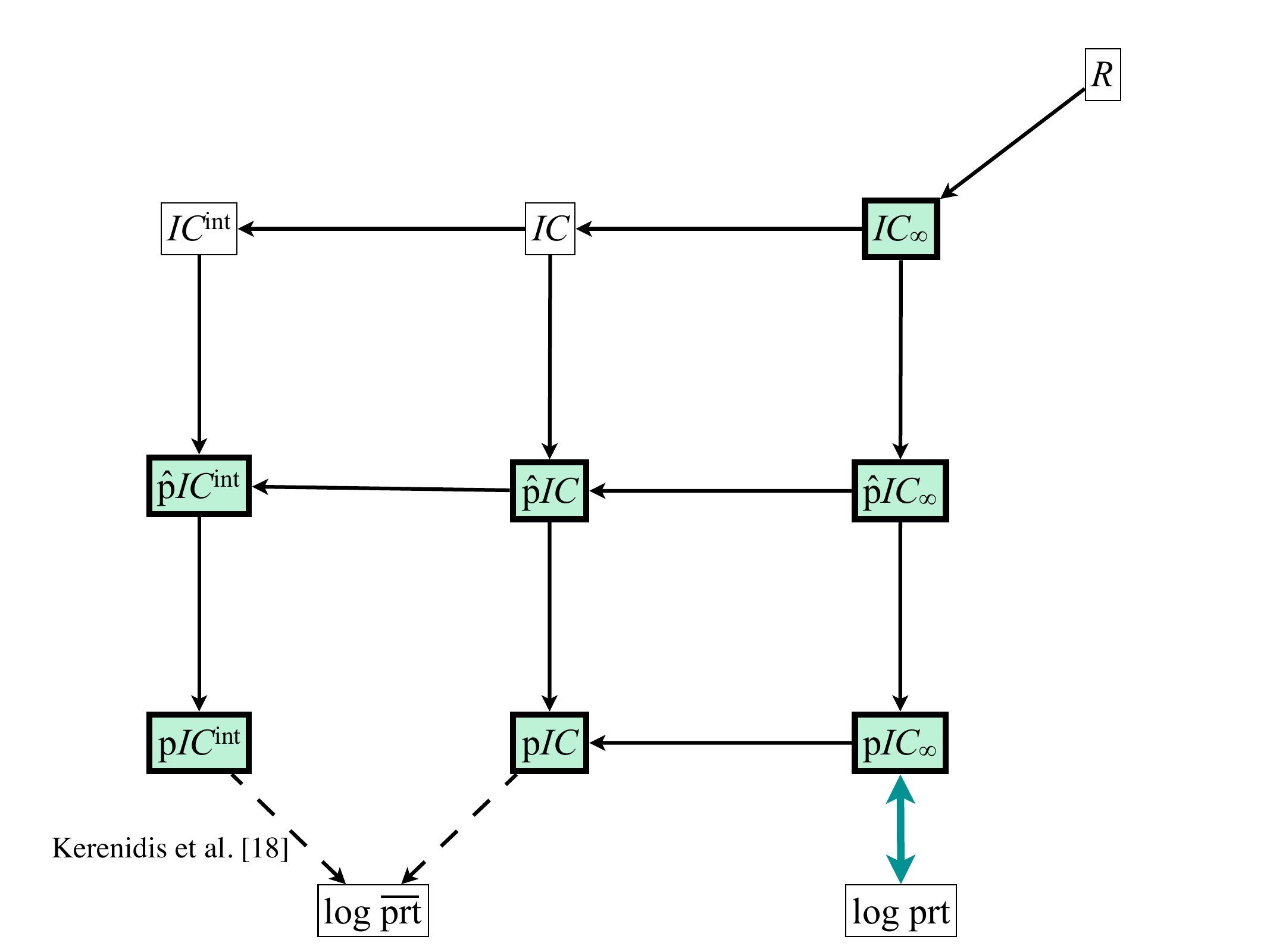}
\caption{An extended version of the map in \Figureref{map}, including
the complexity measures in \Sectionref{extensions}.\label{fig:map-ext}}
\end{center}
\end{figure}

Figure~\ref{fig:map-ext} shows the relationship between the different
complexities we consider. Since, for any protocol, its transcript is a
refined pseudotranscript and refined pseudotranscripts are also
pseudotranscripts, we have 
$\p\mathtt{X}(f,\eps) \leq \xp\mathtt{X}(f,\eps) \leq \mathtt{X}(f,\eps)$,
where $\mathtt{X}$ can be $\ICr, \IC$ or $\ICint$.
Furthermore, analogous to
$\ICint(f,\eps) \leq \IC(f,\eps) \leq \ICr(f,\eps)$,
we have
$\xpICint(f,\eps) \leq \xpIC(f,\eps) \leq \xpICr(f,\eps).$
Finally, in deriving a lower bound for $\ICint(f,\ceps)$ in terms of $\relprt(f,\ceps)$ \cite{KerenidisLaLeRoXi12} only relies on the fact that the transcript (along with the public-coins) $\Pi$ satisfies the factorization condition. Hence, the lower bound of~\cite{KerenidisLaLeRoXi12} holds with $\ICint$ replaced by $\pICint$.

\section*{Acknowledgments}
We gratefully acknowledge Mark Braverman, Prahladh Harsha, Rahul Jain, Anup
Rao and the anonymous referees for helpful suggestions and pointers. The
research was supported in part by NSF grant 1228856 for the first author, by
ITRA, Media Lab Asia, India and by a Ramanujan fellowship from DST, India
for the second author.

\bibliographystyle{plain}
\bibliography{bib}

\begin{thebibliography}{10}

\bibitem{Ablayev96}
Farid~M. Ablayev.
\newblock Lower bounds for one-way probabilistic communication complexity and
  their application to space complexity.
\newblock {\em Theor. Comput. Sci.}, 157(2):139--159, 1996.

\bibitem{Bar-YossefJaKuSi04}
Ziv Bar-Yossef, T.~S. Jayram, Ravi Kumar, and D.~Sivakumar.
\newblock An information statistics approach to data stream and communication
  complexity.
\newblock {\em J. Comput. Syst. Sci.}, 68(4):702--732, 2004.

\bibitem{BarakBrChRa13}
Boaz Barak, Mark Braverman, Xi~Chen, and Anup Rao.
\newblock How to compress interactive communication.
\newblock {\em SIAM J. Comput.}, 42(3):1327--1363, 2013.

\bibitem{Braverman12}
Mark Braverman.
\newblock Interactive information complexity.
\newblock In {\em STOC}, pages 505--524, 2012.

\bibitem{BravermanRa11}
Mark Braverman and Anup Rao.
\newblock Information equals amortized communication.
\newblock In {\em FOCS}, pages 748--757, 2011.

\bibitem{BravermanWe12}
Mark Braverman and Omri Weinstein.
\newblock A discrepancy lower bound for information complexity.
\newblock In {\em APPROX-RANDOM}, pages 459--470, 2012.

\bibitem{ChakrabartiKoWa12}
Amit Chakrabarti, Ranganath Kondapally, and Zhenghui Wang.
\newblock Information complexity versus corruption and applications to
  orthogonality and gap-hamming.
\newblock In {\em APPROX-RANDOM}, pages 483--494, 2012.

\bibitem{ChakrabartiShWiYa01}
Amit Chakrabarti, Yaoyun Shi, Anthony Wirth, and Andrew Chi-Chih Yao.
\newblock Informational complexity and the direct sum problem for simultaneous
  message complexity.
\newblock In {\em FOCS}, pages 270--278, 2001.

\bibitem{FontesJaKeLaLaRo15}
Lila Fontes, Rahul Jain, Iordanis Kerenidis, Sophie Laplante, Mathieu
  Laurière, and Jérémie Roland.
\newblock Relative discrepancy does not separate information and communication
  complexity.
\newblock pages 506--516, 2015.

\bibitem{GanorKoRa14FOCS}
Anat Ganor, Gillat Kol, and Ran Raz.
\newblock Exponential separation of information and communication.
\newblock In {\em FOCS}, pages 176--185, 2014.

\bibitem{GanorKoRa15}
Anat Ganor, Gillat Kol, and Ran Raz.
\newblock Exponential separation of communication and external information.
\newblock In {\em Electronic Colloquium on Computational Complexity (ECCC)},
  2015.

\bibitem{HarshaJaMcRa10}
Prahladh Harsha, Rahul Jain, David McAllester, and Jaikumar Radhakrishnan.
\newblock The communication complexity of correlation.
\newblock {\em {IEEE} Transactions on Information Theory}, 56(1):438--449,
  2010.

\bibitem{HoVe15}
Siu-Wai Ho and Sergio Verd{\'u}.
\newblock Convexity/concavity of {R\'enyi} entropy and $\alpha$-mutual
  information.
\newblock In {\em Information Theory (ISIT), 2015 IEEE International Symposium
  on}, pages 745--749, 2015.

\bibitem{JainKl10}
Rahul Jain and Hartmut Klauck.
\newblock The partition bound for classical communication complexity and query
  complexity.
\newblock In {\em IEEE Conference on Computational Complexity}, pages 247--258,
  2010.

\bibitem{JainRaSe03}
Rahul Jain, Jaikumar Radhakrishnan, and Pranab Sen.
\newblock A direct sum theorem in communication complexity via message
  compression.
\newblock In {\em ICALP}, pages 300--315, 2003.

\bibitem{JainRaSe05}
Rahul Jain, Jaikumar Radhakrishnan, and Pranab Sen.
\newblock Prior entanglement, message compression and privacy in quantum
  communication.
\newblock In {\em IEEE Conference on Computational Complexity}, pages 285--296,
  2005.

\bibitem{JayramKuSi03}
T.~S. Jayram, Ravi Kumar, and D.~Sivakumar.
\newblock Two applications of information complexity.
\newblock In {\em STOC}, pages 673--682, 2003.

\bibitem{KerenidisLaLeRoXi12}
Iordanis Kerenidis, Sophie Laplante, Virginie Lerays, J{\'e}r{\'e}mie Roland,
  and David Xiao.
\newblock Lower bounds on information complexity via zero-communication
  protocols and applications.
\newblock In {\em FOCS}, pages 500--509, 2012.

\bibitem{KushilevitzNi97book}
Eyal Kushilevitz and Noam Nisan.
\newblock {\em Communication complexity}.
\newblock Cambridge University Press, New York, 1997.

\bibitem{PonzioRaVe01}
Stephen~J Ponzio, Jaikumar Radhakrishnan, and Srinivasan Venkatesh.
\newblock The communication complexity of pointer chasing.
\newblock {\em Journal of Computer and System Sciences}, 62(2):323--355, 2001.

\bibitem{PrabhakaranPr14a}
Manoj Prabhakaran and Vinod~M. Prabhakaran.
\newblock Tension bounds for information complexity.
\newblock {\em CoRR}, abs/1408.6285, 2014.

\bibitem{Renyi60}
Alfred R{\'e}nyi.
\newblock On measures of information and entropy.
\newblock In {\em Proceedings of the 4th Berkeley Symposium on Mathematics,
  Statistics and Probability}, pages 547--561, 1960.

\bibitem{SaksSu02}
Michael~E. Saks and Xiaodong Sun.
\newblock Space lower bounds for distance approximation in the data stream
  model.
\newblock In {\em STOC}, pages 360--369, 2002.

\bibitem{Sibson69}
R.~Sibson.
\newblock Information radius.
\newblock {\em Z. Wahrscheinlichkeitstheorie und Verw. Geb.}, 14:149--161,
  1969.

\bibitem{Verdu15}
Sergio Verd\'u.
\newblock $\alpha$-mutual information.
\newblock In {\em Information Theory and Applications Workshop (ITA)}, 2015.

\bibitem{Yao79}
Andrew Chi-Chih Yao.
\newblock Some complexity questions related to distributive computing
  (preliminary report).
\newblock In {\em STOC}, pages 209--213, 1979.

\bibitem{ZakaiZi75}
Moshe Zakai and Jacob Ziv.
\newblock A generalization of the rate-distortion theory and applications.
\newblock In {\em Information Theory New Trends and Open Problems}, pages
  87--123. Springer, 1975.

\bibitem{ZivZa73}
Jacob Ziv and Moshe Zakai.
\newblock On functionals satisfying a data-processing theorem.
\newblock {\em Information Theory, IEEE Transactions on}, 19(3):275--283, 1973.

\end{thebibliography}

\begin{appendices}

\section{Omitted Proofs}
\subsection{$I(A;B)\leq \mininfo(A;B)$}
\label{app:renyi-monotonicity}
For the sake of completeness, we include a proof that $I(A;B)\leq \mininfo(A:B)$.
\begin{align*}
\mininfo(A:B) &= \log \left( \sum_{b\in\cB} \max_{a\in\cA: \prob_A(a) > 0} \prob_{B|A}(b|a) \right)\\
 &\geq \log \left( \sum_{b\in\cB:\prob_B(b)>0} \prob_B(b) \max_{a\in\cA: \prob_A(a) > 0} \frac{\prob_{B|A}(b|a)}{\prob_B(b)} \right)\\
 &\geq \log \left( \sum_{b\in\cB:\prob_B(b)>0} \prob_B(b) \sum_{a\in\cA: \prob_A(a) > 0} \prob_{A|B}(a|b)\frac{\prob_{B|A}(b|a)}{\prob_B(b)} \right)\\
 &= \log \left( \sum_{a\in\cA,b\in\cB:\prob_A(a)>0,\prob_B(b)>0} \prob_{A,B}(a,b)\frac{\prob_{B|A}(b|a)}{\prob_B(b)} \right)\\
 &\geq \sum_{a\in\cA,b\in\cB:\prob_A(a)>0,\prob_B(b)>0} \prob_{A,B}(a,b) \log \left( \frac{\prob_{B|A}(b|a)}{\prob_B(b)} \right)\\
 &= I(A;B).
\end{align*}

\subsection{Proof of \Lemmaref{pICr-le-prt}}
\label{app:pICr-le-prt}
\begin{proof}
Consider the weight function $w:\cT\rightarrow [0,1]$ 
that satisfies the conditions \Eqref{prt-total} and \Eqref{prt-err}
such that $\prt(f,\eps) = \sum_{t\in\cT} w(t)$. Define the random variable
$Q$ over $\cQ=\cT$ such that $\prob_{Q|XY}(t|x,y) = w(t)$ if
$(x,y)\in t$ and 0 otherwise. Note that this is a valid probability
distribution since for all $(x,y)\in\cX\times\cY$, we have
\[ \sum_{t\in\cQ} \prob_{Q|XY}(t|x,y) = \sum_{t\in\cQ:(x,y)\in t} w(t) = 1.\]
Let $a_t,b_t \ge 0$ be such that $a_t\cdot b_t = w(t)$ (for instance,
$a_t=b_t=\sqrt{w(t)}$), and define functions
$\alpha:\cQ\times\cX\rightarrow\Real^+$ and 
$\beta:\cQ\times\cY\rightarrow\Real^+$ as follows:
\begin{align*}
\alpha(t,x) = 
\begin{cases}
a_t & \text{if } x\in \cX_t \\
0 & \text{otherwise}
\end{cases}
&&
\beta(t,y) = 
\begin{cases}
b_t & \text{if } y\in \cY_t \\
0 & \text{otherwise}
\end{cases}
\end{align*}
Then, $\prob_{Q|XY}(t|x,y)=\alpha(t,x)\cdot\beta(t,y)$, and hence it
satisfies the factorization condition. Further, for each
$(x,y)\in\cX\times\cY$, \[ \err_{f,Q}(x,y) = \sum_{t\in\cQ : z_t \not\in
f(x,y)} \prob_{Q|XY}(t|x,y)
=  \sum_{t\in\cQ: (x,y) \in t, z_t \not\in f(x,y)} w(t) \le \eps(x,y). \]

Hence $\pICr(f,\eps) \le \mininfo(X,Y:Q)$. On the other hand,
\[ \mininfo(X,Y:Q) = \log \sum_{t\in\cQ} \max_{x,y} \prob_{Q|XY}(t|x,y) =
\log \sum_{t\in\cT} w(t) = \log \prt(f,\eps),\]
concluding the proof.
\end{proof}

\subsection{Proof of \Theoremref{pIC-relaxprt}}
\label{app:pIC-relaxprt}

\begin{proof}
We shall show that for any distribution $\prob_{XY}=\mu$ over $\cX\times\cY$, 
and any pseudotranscript $Q$ such that $\err_{f,Q} \le \ceps$ (i.e.,
$\forall (x,y)\in\cX\times\cY$, $\err_{f,Q}(x,y) \le \ceps$), $I(X,Y;Q) \ge
\delta \log \relprt^\mu(f,\ceps+\delta) - (\delta\log\log|\cX||\cY|+3)$.
This gives the desired result, since
\[ \pIC(f,\ceps) = \inf_{Q:\err_{f,Q}\leq \ceps} \max_{\prob_{XY}} I(X,Y;Q)
\ge \max_{\prob_{XY}} \inf_{Q:\err_{f,Q}\leq \ceps} I(X,Y;Q) \]
and as shown in \cite{KerenidisLaLeRoXi12}, $\relprt(f,\ceps') = \max_\mu \relprt^\mu(f,\ceps')$.

The proof uses the construction from the proof of \Lemmaref{pICr-ge-prt},
and modifies it carefully. Specifically, we define $\prob_{T|Q,X,Y}$ and
$\omega_{q,t}$ as before. 
(Note that since we are now given a
distribution $\mu$ for the random variables $(X,Y)$, this also gives us a
full distribution $\prob_{Q,T,X,Y}$; below $p(x,y)=\mu(x,y)$.)
Recall that we originally defined $w$ as
$w(t) = \sum_{q\in\cQ} \omega_{q,t}$. Our plan now
is to remove some of the weight on the tiles so that the log of the sum can
be bounded by (roughly) $I(X,Y;Q)/\delta$ as opposed to $\mininfo(X,Y:Q)$.
Towards this, we shall define a set \cB of ``bad'' pairs
$(q,t)\in\cQ\times\cT$ whose weights $\omega_{q,t}$ will not be counted
towards $w'(t)$:
\begin{align*}
w'(t) &= \sum_{(q,t)\in(\cQ\times\cT) \setminus \cB} \omega_{q,t},\qquad  \forall t\in\cT.
\end{align*}
While defining \cB, we need to ensure that the weight removed increases the
{\em average} error $\avgerr^\mu_{f,w'}$ by at most $\delta$ compared to
$\avgerr^\mu_{f,w}=\avgerr^\mu_{f,Q}=\ceps$.

We define parameters $\Delta = (I(XY;Q)+1)/\delta$ and 
for each $q\in\cQ$, $\theta_q = p(q)2^\Delta$.
Let $\hat\alpha(q,t) = \min_{(x,y)\in t} \alpha(q,x)$ and 
$\hat\beta(q,t) = \min_{(x,y)\in t} \beta(q,y)$. Then we define
\begin{align*}
\cB &= \{ (q,t) \in \cQ \times \cT \mid \hat\alpha(q,t)\cdot\hat\beta(q,t) \ge \theta_q. \}
\end{align*}
We make the following claims, which we prove in \Appendixref{missingmass} and \Appendixref{tilebound}.
\begin{claim}
\label{clm:missingmass}
$\sum_{(q,t)\in\cB} p(q,t) \le \delta$.
\end{claim}
\begin{claim}
\label{clm:tilebound}
$\log \sum_{(q,t)\notin\cB} \omega_{q,t} \le \Delta + \log\log (|\cX||\cY|) + 2$.
\end{claim}
Using these claims, we complete the proof.
Firstly, note that $w'(t) \le w(t)$ for every $t\in\cT$ and, since
$w$ satisfies  condition \Eqref{prt-total}, $w'$ satisfies condition \Eqref{relprt-total}.
Also, from \Claimref{missingmass} it follows that 
\begin{align*}
\avgerr^\mu_{f,w'} 
&= 1-\sum_{x,y} p(x,y) \sum_{\substack{t\in \cT: (x,y)\in t,\\z_t \in f(x,y)}} w'(t)
 = 1-\sum_{x,y} p(x,y) \sum_{\substack{(q,t)\in(\cQ\times\cT)\setminus\cB:\\ (x,y)\in t,\\z_t \in f(x,y)}} \omega_{q,t}\\
&= 1-\sum_{x,y} p(x,y) \sum_{\substack{(q,t)\in\cQ\times\cT:\\ (x,y)\in t,\\z_t \in f(x,y)}} \omega_{q,t}
+\sum_{x,y} p(x,y) \sum_{\substack{(q,t)\in\cB:\\ (x,y)\in t,\\z_t \in f(x,y)}} \omega_{q,t}\\
&=   \avgerr^\mu_{f,w} + \sum_{(q,t)\in\cB} \sum_{\substack{(x,y)\in t,\\z_t \in f(x,y)}} p(x,y) \omega_{q,t}
 \;\;\le\;\; \avgerr^\mu_{f,w} + \sum_{\substack{(q,t)\in\cB}} \sum_{(x,y)\in t} p(x,y) \omega_{q,t}\\
&= \avgerr^\mu_{f,w} + \sum_{\substack{(q,t)\in\cB}} \sum_{(x,y)\in \cX\times\cY} p(x,y) p(q,t|x,y) 
	&\text{ by \Eqref{omegaqt-pqtxy}}\\
&=   \avgerr^\mu_{f,w} + \sum_{\substack{(q,t)\in\cB}} p(q,t) \le \ceps +
\delta 
	&\text{ by \Claimref{missingmass}}
\end{align*}
Hence,
\begin{align*}
\log \relprt^\mu(f,\ceps+\delta) &\le \sum_{t\in\cT} w'(t) = \log \sum_{(q,t)\notin\cB} \omega_{q,t} \\
& \le \Delta + \log\log |\cX||\cY| + 2 
	&\text{by \Claimref{tilebound}}\\
&= \frac{I(X,Y;Q)}{\delta} +\frac1{\delta} + \log\log|\cX||\cY| + 2 \\
&\le \frac{I(X,Y;Q)}{\delta} + \log\log|\cX||\cY| + \frac3\delta  & \text{ since } \delta \in [0,1]
\end{align*}
That is, $I(X,Y;Q) \geq \delta\log \relprt^\mu(f,\ceps+\delta) + (\delta \log\log|\cX||\cY| + 3)$,
as was required to prove.
\end{proof}

The proofs of the two claims used above follow.

\subsubsection{Proof of \Claimref{missingmass}}
\label{app:missingmass}

\begin{proof}
This claim follows from Markov's inequality applied to an appropriate
random variable, whose mean is related to $I(XY;Q)$. First, we
expand $I(XY;Q)$ as follows:
\begin{align*}
I(XY;Q) 
&= \sum_{q\in\cQ,x\in\cX,y\in\cY} p(q,x,y) \log \frac{p(q|x,y)}{p(q)} \\
&= \sum_{q\in\cQ,t\in\cT,x\in\cX,y\in\cY} p(q,t,x,y) \log \frac{p(q|x,y)}{p(q)} \\
&= \sum_{q\in\cQ,t\in\cT} p(q,t) \sum_{(x,y)\in t} p(x,y|q,t) \log \frac{p(q|x,y)}{p(q)} 
	&\text{since } (x,y)\notin t \implies p(q,t,x,y)=0 \\
&= \sum_{q\in\cQ,t\in\cT} p(q,t) \varphi(q,t)
\end{align*}
where we have defined 
\begin{align*}
\varphi(q,t) = \begin{cases}
\sum_{(x,y)\in t} p(x,y|q,t) \log \frac{p(q|x,y)}{p(q)} & \text{if } p(q,t)\not=0 \\
0 & \text{otherwise}
\end{cases}
\end{align*}
That is, $\varphi(q,t)$ is the average value of $\log \frac{p(q|x,y)}{p(q)}$
averaged over all $(x,y)\in t$ using the distribution $\prob_{XY|Q=q,T=t}$.
We note that for all $(q,t)\in\cB$,
$\varphi(q,t) \ge \Delta$, since for each 
$(x,y)\in t$, $p(q|x,y)=\alpha(q,x)\beta(q,y) 
\ge \hat\alpha(q,t)\hat\beta(q,t) \ge \theta_q$
and hence $\log \frac{p(q|x,y)}{p(q)} \ge \log \frac{\theta_q}{p(q)} =
\Delta$. This suggests the possibility of using the Markov inequality to bound
$\sum_{(q,t)\in\cB} p(q,t)$. However, $\varphi(q,t)$ could be negative, and we cannot directly use the
above expression for $I(X,Y;Q)$ in a Markov inequality. However, we claim that
removing the negative terms from $\sum_{q,t} p(q,t)\varphi(q,t)$
does not increase the sum significantly, which will let us 
still apply the Markov inequality.

More precisely, let $\cD=\{ (q,t) \in \cQ\times\cT \mid
\min_{(x,y)\in t} p(q|x,y) \ge p(q) \}$. 
Note that if $(q,t)\in\cD$, then $\varphi(q,t) \ge 0$.
We claim that 
\begin{align}
\label{eq:I-smallloss}
I(X,Y;Q) \ge  \left( \sum_{(q,t)\in\cD} p(q,t)\varphi(q,t) \right) - 1.
\end{align}
Assuming \Eqref{I-smallloss}, we can conclude the proof of the claim as
follows. Note that $\cB \subseteq \cD$ since if $(q,t)\in\cB$,
$\min_{(x,y)\in t} p(q|xy) = \hat\alpha(q,t)\cdot\hat\beta(q,t) \ge \theta_q
\ge p(q)$. Also, recall that for $(q,t)\in\cB$, $\varphi(q,t) \ge \Delta$. Hence, 
\[ \delta\Delta = I(X,Y;Q) + 1 \ge \sum_{(q,t)\in \cD} p(q,t)\varphi(q,t) \ge \Delta\sum_{(q,t)\in\cB} p(q,t),\]
and therefore $\sum_{(q,t)\in\cB} p(q,t) \le \delta$.

To prove \Eqref{I-smallloss}, consider again the expansion of $I(X,Y;Q)$
as
\[ 
I(XY;Q) 
= \left( \sum_{\substack{q\in\cQ,x\in\cX,y\in\cY:\\ p(q|x,y) \ge p(q)}} p(q,x,y) \log \frac{p(q|x,y)}{p(q)} \right)
- \left( \sum_{\substack{q\in\cQ,x\in\cX,y\in\cY:\\ p(q|x,y) < p(q)}} p(q,x,y) \log \frac{p(q)} {p(q|x,y)} \right).\]
in which all the terms within each summation is non-negative. 
To bound the second term, writing $\eta =
\sum_{q,x,y: p(q|x,y) < p(q)} p(q,x,y)$, we use Jensen's inequality to write
\begin{align*}
\sum_{\substack{q\in\cQ,x\in\cX,y\in\cY:\\ p(q|x,y) < p(q)}} p(q,x,y) \log \frac{p(q)} {p(q|x,y)}
&\le 
\eta \log \sum_{\substack{q\in\cQ,x\in\cX,y\in\cY:\\ p(q|x,y) < p(q)}} \frac{p(q,x,y)}{\eta} \cdot \frac{p(q)} {p(q|x,y)}\\
&=   \eta \log \frac1\eta + \eta \log \sum_{\substack{q\in\cQ,x\in\cX,y\in\cY:\\ p(q|x,y) < p(q)}} p(x,y)p(q)\\
&\le \eta \log \frac1\eta  \le \frac{\log e}e < 1.
\end{align*}
where to get to the last line we used the fact that
$\displaystyle \sum_{\substack{q\in\cQ,x\in\cX,y\in\cY:\\ p(q|x,y)<p(q)}} p(x,y)p(q) \le \sum_{q\in\cQ,x\in\cX,y\in\cY} p(x,y)p(q) = 1$.
Hence 
\begin{align*}
I(XY;Q) 
&\ge \left(\sum_{\substack{q\in\cQ,x\in\cX,y\in\cY:\\ p(q|x,y) \ge p(q)}} p(q,x,y) \log \frac{p(q|x,y)}{p(q)} \right)- 1 \\
&\ge \left(\sum_{(q,t)\in\cD,(x,y)\in t} p(q,t,x,y) \log \frac{p(q|x,y)}{p(q)} \right) - 1 
 &\text{ since } (x,y)\in t, (q,t) \in \cD \implies p(q|x,y) \ge p(q)\\
&= \left(\sum_{(q,t)\in\cD} p(q,t) \varphi(q,t) \right) - 1
\end{align*}
completing the proof of \Eqref{I-smallloss} and of the claim.
\end{proof}

\subsubsection{Proof of \Claimref{tilebound}}
\label{app:tilebound}

\begin{proof}
We need to upper-bound
\[ \sum_{\substack{q\in\cQ,t\in\cT:\\(q,t)\notin\cB}} \omega_{q,t} =
\sum_{q\in\cQ} \sum_{\substack{t\in\cT_q:\\(q,t)\notin\cB}} \sigma_{q,t}\tau_{q,t}. \]
For this we shall use a geometric interpretation of this sum.

Fix $q\in\cQ$. Recall from the proof of \Claimref{telescope},
that for each $q$, we order $\cX = \{ x_1,\cdots,x_M \}$ and
$\cY = \{ y_1,\cdots,y_N \}$ such that $\alpha(q,x_i) \ge \alpha(q,x_{i-1})$ 
and $\beta(q,y_j) \ge \beta(q,y_{j-1})$ (taking
$\alpha(q,x_0)=\beta(q,y_0)=0$), and
$t_{ij} = (\cX_i \times \cY_j,z_q)$ for
$(i,j)\in[M]\times[N]$, where $\cX_i = \{x_i,\cdots,x_M\}$,
$\cY_j=\{y_j,\cdots,y_N\}$. Then
\[\cT_q = \{ t_{ij} \mid (i,j)\in[M]\times[N], \alpha(q,x_i) >
\alpha(q,x_{i-1}), \beta(q,y_j) > \beta(q,y_{j-1}) \}. \]
Consider the rectangular region $R \subseteq \Real^2$ defined
by the diagonally opposite vertices $(0,0)$ and $(\alpha_q^*,\beta_q^*)$,
where $\alpha_q^* = \max_{x\in\cX} \alpha(q,x)$ and
$\beta_q^* = \max_{y\in\cY} \beta(q,y)$. For each $(i,j)\in[M]\times[N]$
let the (possibly empty) rectangular region $R_{ij}$ be defined by
opposite vertices $(\alpha(q,x_{i-1}),\beta(q,y_{j-1}))$ and $(\alpha(q,x_{i}),\beta(q,y_{j}))$.
(See \Figureref{tilebound}.)
Then note that the entire region $R$ is tiled by the rectangles $R_{ij}$,
without any overlap:
\begin{align*}
R = \bigcup_{(i,j)\in[M]\times[N]} R_{ij} && (i,j)\not=(i',j') \implies
R_{ij} \cap R_{i'j'} = \emptyset.
\end{align*}
Further, the area of the rectangle $R_{ij}$ is the same as
$\omega_{q,t_{ij}} =  \sigma_{q,t_{ij}}\tau_{q,t_{ij}} = \left({\alpha(q,x_i)} - {\alpha(q,x_{i-1})}\right)\left({\beta(q,y_j)} - {\beta(q,y_{j-1})}\right)$. 
Thus,
\[ \sum_{\substack{t\in\cT_q:\\(q,t)\notin\cB}} \omega_{q,t} =
\sum_{\substack{(i,j)\in[M]\times[N]:\\(q,t_{ij})\notin\cB}} \mathrm{area}(R_{ij}).\]

Now we need to identify the rectangles $R_{ij}$ such that
$(q,t_{ij})\notin\cB$. Firstly, recall that $\hat\alpha(q,t_{ij}) =
\min_{(x,y)\in t_{ij}} \alpha(q,x) = \alpha(q,x_i)$, and similarly
$\hat\beta(q,t_{ij})=\beta(q,y_j)$. Hence $(q,t_{ij})\in\cB$ if and only if
$\alpha(q,x_i)\beta(q,y_j) \ge \theta_q$. In terms of the rectangle $R_{ij}$
this corresponds to having its top-right vertex (i.e.,
$(\alpha(q,x_i),\beta(q,y_j))$) fall ``above'' the hyperbola defined by the
equation $xy=\theta_q$. Thus if $(q,t_{ij})\notin\cB$, then the entire
rectangle $R_{ij}$ is below the hyperbola $xy=\theta_q$. The sum of their
areas is upper-bounded by the area within $R$ that is under this hyperbola.

We consider two cases for $q$: when the hyperbola
intersects $R$ and when it does not; the latter happens when $\theta_q >
\alpha_q^*\beta_q^*$. Let $\cS = \{ q \mid \theta_q > \alpha_q^*\beta_q^* \}$.
If $q\in\cS$, then clearly the area of $R$ below the hyperbola is
the entire area, $\alpha_q^*\beta_q^*$. Otherwise, the area under the
hyperbola is found by integration as
\[ \theta_q + \int_{\frac{\theta_q}{\beta_q^*}}^{\alpha_q^*}
\frac{\theta_q}{x} dx = \theta_q + \theta_q \ln \frac{\alpha_q^*\beta_q^*}{\theta_q},\]
where $\ln$ stands for natural logarithm. 

Let $\lambda = \sum_{q\notin\cS} p(q)$. Then,
\begin{align*}
\sum_{\substack{(q,t)\in(\cQ\times\cT)\setminus\cB:\\q\in\cS}} \omega_{q,t}
&= \sum_{q\in\cS} \alpha_q^*\beta_q^* \le \sum_{q\in\cS} \theta_q = (1-\lambda) 2^\Delta \\
\sum_{\substack{(q,t)\in(\cQ\times\cT)\setminus\cB:\\q\notin\cS}} \omega_{q,t}
&\le \sum_{q\in\cQ\setminus\cS} \theta_q + \theta_q \ln \frac{\alpha_q^*\beta_q^*}{\theta_q} \\
&= \lambda2^\Delta 
	+ \lambda2^\Delta\sum_{q\in\cQ\setminus\cS} \frac{p(q)}{\lambda}\ln \frac{\alpha_q^*\beta_q^*}{p(q)2^\Delta} \\
&\le \lambda2^\Delta
	+ \lambda2^\Delta\ln \sum_{q\in\cQ\setminus\cS} \frac{\alpha_q^*\beta_q^*}{\lambda2^\Delta} 
	& \text{By Jensen's inequality}\\
&\le \lambda2^\Delta
	+ \lambda2^\Delta \ln \left( \sum_{q\in\cQ} \alpha_q^*\beta_q^* \right) 
	+ \lambda2^\Delta \ln\frac1{\lambda2^\Delta} \\
&\le \lambda2^\Delta + 2^\Delta \cdot \mininfo(X,Y:Q) \cdot \ln 2 + \frac1e
	& \text{since for all } a>0,  a\ln\frac1a \le \frac1e \\
&\le \lambda2^\Delta + 2^\Delta \cdot \log|\cX||\cY| \cdot \ln 2 + \frac1e &\text{since }  \mininfo(X,Y:Q)\leq\log |\cX||\cY|\\
\sum_{(q,t)\in(\cQ\times\cT)\setminus\cB} \omega_{q,t}
	&\le 2^\Delta (1 +  \log|\cX||\cY| \cdot \ln 2 + \frac1e) \\
	&\le 2^\Delta (4\log|\cX||\cY|) &\text{since } |\cX||\cY| \ge 2
\end{align*}

Note that we assumed $|\cX||\cY|\ge 2$, because otherwise $|\cX|=|\cY|=1$
and the theorem holds trivially (with LHS being 0 and RHS being negative).
From the above we obtain that $\log
\sum_{(q,t)\in(\cQ\times\cT)\setminus\cB} \omega_{q,t} \le \Delta + 
\log\log|\cX||\cY| + 2$ completing the proof of the claim.
\end{proof}
\end{appendices}

\end{document}